\documentclass[10pt,a4paper]{article}
\usepackage[utf8]{inputenc}
\usepackage[english]{babel}
\usepackage{amsmath}
\usepackage{amsfonts}
\usepackage{amssymb}
\usepackage{graphicx}
\usepackage{float}
\usepackage{fixltx2e}
\usepackage{amsthm}
\usepackage{bigints}
\usepackage{xcolor}
\usepackage{float}
\renewcommand{\qedsymbol}{$\square$}

\newtheorem{theorem}{Theorem}[section]
\newtheorem{lemma}[theorem]{Lemma}

\theoremstyle{definition}

\newtheorem{rem}{Remark}
\usepackage{authblk}
\author[*]{Kasturi Das}
\author[*]{S. Krishnaswamy}
\author[*]{S. Majhi}
\affil[*]{EEE Department, IIT Guwahati}

\begin{document}

\title{H$_2$-optimal model order reduction over a finite time interval}

\maketitle
\begin{abstract}
For a time-limited version of the H$_2$ norm defined over a fixed time interval, we obtain a closed form expression of the gradients.  After that, we use the gradients to propose a time-limited model order reduction method. The method involves obtaining a reduced model which minimizes the time-limited H$_2$ norm, formulated as a nonlinear optimization problem. The optimization problem is solved using standard optimization software. 
\end{abstract}

\section{Introduction}
Capturing system dynamics accurately requires large- scale, linear dynamical models. Simulating or analysing such models and designing controllers require considerable computational effort. Such issues are resolved by replacing the large model with a lower order approximation based on various performance measures. Model order reduction  techniques are used in a wide range of areas including computational aerodynamics, large-scale network systems, microelectronics, electromagnetic systems, chemical processes etc. \cite{ModelReduction3}. References \cite{ModelReduction1}, \cite{ModelReduction2} contain a comprehensive discussion on a large number of model reduction techniques available in literature.

 The H$_2$ norm of the error between the original and the reduced system acts as an important performance measure for obtaining reduced order models. In H$_2$ optimal model reduction the aim is to find a lower order model that minimizes this norm. Since finding global minimizers is a difficult task, the existing methods focus on finding local minimizers. These methods are divided into two categories: optimization-based methods and tangential interpolation methods. In optimization based methods, the task of model reduction is formulated as an optimization problem over various manifolds \cite{wilson1970optimum,yan1999approximate,sato2015riemannian,sato2017riemannian}. The solution yields optimal reduced models. Tangential interpolation methods use Krylov-based algorithms and work well for large-scale systems. Examples of such methods include Iterative Rational Krylov Algorithm (IRKA) \cite{gugercin2008h_2} and Two-Sided Iteration Algorithm (TSIA) \cite{xu2011optimal}. Reference \cite{petersson2014model} deals with optimization based frequency-limited H$_2$ optimal model reduction. 

 Availability of simulation data for a finite time interval or the need to approximate the system behaviour over a finite time interval led to the development of finite time model reduction methods. These include methods such as Proper Orthogonal Decomposition (POD) \cite{holmes2012turbulence}, Time-Limited Balanced Truncation (TL-BT) \cite{gawronski1990model} etc. Error bounds for TL-BT are proposed in \cite{redmann2018output, redmann2020lt2}.
References \cite{kurschner2018balanced} and \cite{duff2021numerical} deal with implementation of TL-BT for large-scale continuous and discrete systems respectively. Lyapunov based time-limited H$_2$ optimality conditions are obtained in \cite{goyal2019time} using a time-limited H$_2$ norm. The same paper proposes an iterative scheme similar to TSIA \cite{xu2011optimal}. We refer to this scheme as TL-TSIA. Projection-based algorithms like TL-BT and TL-TSIA fail to exactly satisfy the time-limited optimality conditions. Reference \cite{sinani2019h2} obtains interpolation based first-order necessary conditions for time-limited H$_2$ optimality and proposes an optimization algorithm named FHIRKA. This algorithm produces time-limited H$_2$ optimal models but is valid for SISO systems.

  To the best of the authors' knowledge, there are no optimization based algorithms in the literature that yield time-limited H$_2$ optimal reduced models for both SISO and MIMO systems. In this letter, we aim to fill this gap.
The time-limited H$_2$ optimal model reduction problem is formulated as an optimization problem. We derive closed-form expressions of the gradients of the objective function. These gradients are used with standard quasi-Newton solvers to propose a time-limited H$_2$ optimal model reduction method. We initialize the proposed method with reduced models obtained from time-limited projection based model reduction techniques. Two numerical examples show how our proposed method significantly improves the objective function compared to the projection based methods for reduced orders less than a certain upper bound. Due to space constraints, we have demonstrated a single example. However, we observe that the bound is different for different models. 
 
 The letter is arranged as follows. In Section \ref{sec:preliminaries} we discuss some basic concepts related to model order reduction over a limited time. In Section \ref{Problem Formulation}, we formulate time-limited H$_2$ optimal model reduction as an  
optimization problem and derive the gradients of the objection function in Section \ref{Gradients}. In Section \ref{Algorithm} we propose a method for solving the optimization problem by using the gradients and discussing its computational complexity in Section \ref{Computation Cost}. The proposed method is implemented on two numerical examples in Section \ref{Example}. We conclude the paper in Section \ref{Conclusion}.

\paragraph*{\textbf{Notations}} Let $\mathbb{R}$ and $\mathbb{C}$ be the set of real and complex numbers respectively. For a matrix $P \in \mathbb{R}^{n \times  n}$, Tr(P) denotes the trace, $P^T$ denotes the transpose, $\left\Vert P \right\Vert$ denotes the 2-norm and $\left\Vert P \right\Vert_{F}$ denotes the Frobenius norm of the matrix $P$. Let us consider a function $f:[0,\infty) \to \mathbb{R}^{p \times m}$ whose Laplace transform $F(s) \in \mathbb{C}^{p \times m}$ exists. The time-limited H$_2$ norm of $F$, denoted by $\left\Vert F \right\Vert_{\text{H}_{2,\tau}}$, is defined as $\sqrt{\int_{0}^{\tau}\left\Vert f(t) \right\Vert_{F}^2dt}$ where $\left\Vert f(t) \right\Vert_{F}^2 = Tr({f(t)}^T f(t)) $. The $L_{\infty}^{\tau}$ norm and $L_{2}^{\tau}$ norm of $f(t)$ are defined as follows,
\begin{align*}
    \left\Vert f \right\Vert_{L_{2}^{\tau}} &= \sqrt{\int_{0}^{\tau} \left\Vert f(t)\right\Vert^2 dt} \\
    \left\Vert f \right\Vert_{L_{\infty}^{\tau}} &= \operatorname*{sup}_{t \in [0,\tau]} \left\Vert f(t) \right\Vert
\end{align*}

\section{Preliminaries}
\label{sec:preliminaries}

A stable and strictly proper Linear Time-Invariant (LTI) system, $\Sigma$ is given by,
\begin{subequations}\label{eq1}
\begin{equation} 
\dot{x}(t) = Ax(t)+Bu(t), \, x(0)= 0,\\
\end{equation}
\begin{equation} 
y(t) = Cx(t), \, t \geq 0   
\end{equation}
\end{subequations}
where $A \in \mathbb{R}^{n \times n}$, $B \in \mathbb{R}^{n \times m}$ and $C \in \mathbb{R}^{p \times n}$. Let $H(s)$ be the transfer function and $h(t)$ be the impulse response. We assume that the state dimension $n$ is large and is much larger than the number of inputs and outputs, i.e.\ $n \gg m,p$. 

%

Over a limited time interval $\begin{bmatrix}0 & \tau \end{bmatrix}$ with $\tau<\infty$, time-limited gramians are defined in \cite{gawronski1990model} as follows,
\begin{align}
P_{\tau} = \int_{0}^{\tau} e^{At}BB^{T}e^{A^{T}t}dt, \quad Q_{\tau} = \int_{0}^{\tau} e^{A^{T}t}C^{T}Ce^{At}dt
\end{align}

The time-limited gramians  are solutions of the following Lyapunov equations 
\begin{align}
AP_{\tau}+P_{\tau}A^{T}+BB^{T}-e^{A\tau}BB^{T}e^{{A^T}\tau} &= 0 \\
A^{T}Q_{\tau}+Q_{\tau}A+C^{T}C-e^{A^{T}\tau}C^{T}Ce^{A\tau} &= 0
\end{align}  


%

The following expressions for $\|G \|^{2}_{\text{H}_{2,\tau}}$ are derived in \cite{goyal2019time}.
\begin{equation}\label{FullorderGramians}
\|G \|^{2}_{\text{H}_{2,\tau}} = \text{Tr } CP_{\tau}C^{T} = \text{Tr } B^{T}Q_{\tau}B 
\end{equation}

Since the H$_{2,\tau}$ norm is defined over a strictly finite time interval, a system need not be asymptotically stable inorder to have a finite H$_{2,\tau}$ norm. 

\section{Optimization based time-limited model reduction}
\label{Problem Formulation}

Consider a reduced order system $\hat{\Sigma}$ given by,
\begin{subequations}\label{eq2}
\begin{equation} 
\dot{x_r}(t) = A_r x_r(t)+B_ru(t), \, x_r(0)= 0,\\
\end{equation}
\begin{equation} 
y_r(t) = C_r x_r(t), \, t \geq 0   
\end{equation}
\end{subequations}
where $A_r \in \mathbb{R}^{r \times r}$, $B_r \in \mathbb{R}^{r \times m}$ and $C \in \mathbb{R}^{p \times r}$. Let $H_r(s)$ be the transfer function and $h_r(t)$ be the impulse response. It is essential that $r \ll n$ and $y-y_r$ is small for an appropriate time limited norm.

For all admissible inputs $u(t)$ with unity $L_{2}^{\tau}$ norm, the following relation holds \cite{goyal2019time}.
\begin{equation} \label{H2tau optimization}
 \left\Vert y(t)-y_r(t) \right\Vert_{L_{\infty}^{\tau}} \leq  \left\Vert G-G_r \right\Vert_{\text{H}_{2,\tau}}
\end{equation}
Thus, minimizing the H$_{2,\tau}$ error norm ensures that $y_r(t)$ is a good approximation of $y(t)$ over the time interval $\begin{bmatrix} 0 & \tau \end{bmatrix}$.

In this paper, we aim to obtain H$_{2,\tau}$ optimal reduced models by solving the following optimization problem:

\begin{equation}\label{H2tau optimal model reduction}
    \left\Vert G -G_r \right\Vert_{\text{H}_{2,\tau}} = \operatorname*{minimize \quad}_{\text{dim}(\hat{G})=r} \left\Vert G-\hat{G} \right\Vert_{\text{H}_{2,\tau}}
\end{equation}
The feasible set for the  optimization problem formulated above comprises of all the reduced order systems of the form (\ref{eq2}) with state dimension $r$. $\left\Vert G-G_r \right\Vert_{\text{H}_{2,\tau}}$ is defined as $\left(\bigintss_{0}^{\tau} \left\Vert Ce^{At}B-C_re^{A_r t}B_r \right\Vert_{F}^2 dt\right)^{\frac{1}{2}}$. \\
The error system $(G-G_r)$ can be represented by the following state-space realization.
\begin{equation}\label{Error System}
\{A_e,B_e,C_e \}=\left\{\begin{bmatrix}A & 0 \\ 0 & A_r \end{bmatrix}, \begin{bmatrix}B \\ B_r \end{bmatrix}, \begin{bmatrix}C & - C_r \end{bmatrix} \right\}
\end{equation}
As a consequence of (\ref{FullorderGramians}), the square of the H$_{2,\tau}$ norm of the above  realization can be expressed as
\begin{equation}\label{Error Cost Function}
\left\Vert G -G_r \right\Vert_{\text{H}_{2,\tau}}^2 = \text{Tr}({B_e}^{T}Q_{e,\tau}B_e)= \text{Tr}({C_e}P_{e,\tau}{C_e}^{T}) 
\end{equation}
Here, $P_{e,\tau}$ and $Q_{e,\tau}$ are the time-limited controllability and observability gramians and they satisfy the following Lyapunov equations,
\begin{subequations}
\begin{equation}\label{Time-limited Controllability Gramian}
    A_e P_{e,\tau}+P_{e,\tau}{A_e}^{T}+B_e{B_e}^{T}-e^{A_e\tau}B_e{B_e}^{T}e^{{A_e}^{T}\tau}=0
\end{equation}
\begin{equation}\label{Time-Limited Observability Gramian}
    {A_e}^{T} Q_{e,\tau}+Q_{e,\tau}{A_e}+{C_e}^{T} C_e-e^{{A_e}^{T}\tau}{C_e}^{T}C_e e^{A_e\tau}=0
    \end{equation}
\end{subequations}
For the realization (\ref{Error System}), the corresponding gramians $P_{e,\tau}$ and $Q_{e,\tau}$ can be partitioned as follows,
\begin{equation}\label{GramianPartition}
    P_{e,\tau} = \begin{pmatrix} P_\tau & X_\tau \\ X_\tau^{T} & P_{r,\tau} \end{pmatrix} \quad Q_{e,\tau} = \begin{pmatrix} Q_\tau & Y_\tau \\ Y_\tau^{T} & Q_{r,\tau} \end{pmatrix}
\end{equation}
Further the matrix $e^{A_e\tau}$ can be partitioned as follows,
\begin{equation}\label{MatrixExpPartition}
e^{A_e\tau} = \begin{bmatrix} e^{A\tau} & 0 \\ 0 & e^{A_r\tau} \end{bmatrix}    
\end{equation}
Substituting the partitions (\ref{GramianPartition}) and (\ref{MatrixExpPartition}) for $P_{e,\tau}$, $Q_{e,\tau}$ and $e^{A_e\tau}$ in equations (\ref{Time-limited Controllability Gramian}) and (\ref{Time-Limited Observability Gramian}) we get the following time-limited Lyapunov and time-limited Sylvester equations,
\begin{subequations}
\begin{equation}\label{Full Controllability Lyapunov Equation}
AP_\tau+P_{\tau}A^{T}+BB^{T}-e^{A\tau}BB^{T}e^{A^{T}\tau} = 0
\end{equation}
\begin{equation}\label{Controllability Sylvester Equation}
AX_\tau+X_\tau {A_r}^{T}+B{B_r}^{T}-e^{A\tau}B{B_r}^{T}e^{{A_r}^{T}\tau} = 0 
\end{equation}
\begin{equation}\label{Reduced Controllability Lyapunov Equation}
A_r P_{r,\tau}+P_{r,\tau}{A_r}^{T}+B_r{B_r}^{T}-e^{{A_r}\tau}{B_r}{B_r}^{T}e^{{A_r}^{T}\tau} = 0
\end{equation}
\begin{equation}\label{Full Observability Lyapunov Equation}
A^{T}Q_\tau+Q_\tau A+C^{T}C-e^{A^{T}\tau}C^{T}Ce^{A\tau} = 0 
\end{equation}
\begin{equation}\label{Observability Sylvester Equation}
A^{T}Y_\tau+Y_\tau A_r-C^{T}C_r+ e^{A^{T}\tau}C^{T}C_re^{A_r\tau} = 0 
\end{equation}
\begin{equation}\label{Reduced Observability Lyapunov Equation}
{A_r}^{T}Q_{r,\tau}+Q_{r,\tau}A_r+{C_r}^{T}{C_r}-e^{A_r^{T}\tau}{C_r}^{T}{C_r}e^{{A_r}\tau} = 0 
\end{equation}
\end{subequations}
Here, $P_{\tau}$ and $Q_{\tau}$ are the controllability and observability gramian respectively for the full order model (\ref{eq1}). $P_{r,\tau}$ and $Q_{r,\tau}$
are the controllability and observability gramian respectively for the reduced model (\ref{eq2})
Additionally, substituting the above partitions we can simplify (\ref{Error Cost Function}) as follows
\begin{subequations}
\begin{equation}\label{ErrorCostFunction1}
\begin{aligned}
    &\left\Vert G -G_r \right\Vert_{\text{H}_{2,\tau}}^2 \\
    &= \text{Tr}\left(C P_\tau C^{T}-2C X_\tau C_r^T+C_r P_{r,\tau} C_r^T\right)
\end{aligned}
\end{equation}
\begin{equation}\label{ErrorCostFunction2}
= \text{Tr}\left( B^{T} Q_\tau B+2B^T Y_\tau B_r + B_r^T Q_{r,\tau} B_r\right)
\end{equation}
\end{subequations}

\subsection{Gradients of the Cost Function}
\label{Gradients}
For a matrix valued function $f:\mathbb{R}^{m \times n} \to \mathbb{R}$, the gradient at $M \in \mathbb{R}^{m \times n}$ is another matrix ${\bigtriangledown} f(M) \in \mathbb{R}^{m \times n}$ which is given by Definition 3.1 of \cite{van2008h2}. The inner product of two matrices is given by $\langle A,B \rangle= \text{Tr}(A^{T}B)$. We now proceed to derive gradients of the objective function (\ref{ErrorCostFunction1} or \ref{ErrorCostFunction2}) with respect to the reduced system matrices. The following lemma from \cite{van2008h2} is essential for proving the subsequent theorem.

\begin{lemma}\label{Lemma1}
 If $AM+MB+C=0$ and $NA+B N+D=0$ then $\text{Tr}(C N)=\text{Tr}(D M)$. 
\end{lemma}

\begin{theorem} \label{Theorem 1}
For the cost function $J = \left\Vert G-G_r \right\Vert_{H_{2,\tau}}^2$,  
the gradients with respect to $A_r$, $B_r$ and $C_r$ denoted by ${\nabla}_{A_r} J$, $ {\nabla}_{B_r} J$ and ${\nabla}_{C_r}J$ respectively are 
\begin{subequations}
\begin{equation}\label{GradA}
    {\nabla}_{A_r} J = 2(Q_{r,\tau}P_r+Y^{T}_\tau X+\tau (L(A_r \tau,S_{\tau})^T)) 
\end{equation}
\begin{equation}\label{GradB}
    {\nabla}_{B_r} J = 2(Q_{r,\tau}B_r+Y^{T}_\tau B)
\end{equation}
\begin{equation}\label{GradC}
 {\nabla}_{C_r} J = 2(C_r P_{r,\tau}-C X_\tau)
\end{equation}
\end{subequations}
where $P_{r,\tau}$, $Q_{r,\tau}$, $X_\tau$, and $Y_\tau$ are solutions of (\ref{Reduced Controllability Lyapunov Equation}), (\ref{Reduced Observability Lyapunov Equation}), (\ref{Controllability Sylvester Equation}), and (\ref{Observability Sylvester Equation}) respectively. $P_r$ and $X$ are obtained by solving the following Lyapunov and Sylvester equations.
\begin{subequations}
\begin{equation} \label{Extra Lyapunov Equation}
    P_r A_r^T + A_r P_r + B_r B_r^T = 0
\end{equation}
\begin{equation} \label{Extra Sylvester Equation}
    X^T A^T + A_r X^T + B_r B^T = 0
\end{equation}
\end{subequations}
Here, the function $L(X,Y)$ is the Fr{\'e}chet derivative of the matrix exponential of $X$ along the  direction $Y$ \cite{al2009computing}. $S_\tau$ is given by
\begin{equation}\label{Stau}
 S_{\tau} = \left(X^Te^{A^T\tau}C^T C_r-P_r e^{A_r^T \tau}C_r^TC_r \right)
 \end{equation}
\end{theorem}
\begin{proof} 
Consider the expression (\ref{ErrorCostFunction2}) of the cost function.
For a perturbation of $\Delta_{A_r}$ in $A_r$, the corresponding first-order perturbation in $J$ denoted by $\Delta_J^{A_r}$ is,
\begin{equation}\label{Differential J}
\begin{aligned}
    \Delta_J^{A_r} 
             &= \text{Tr}\left(2B_r B^T\Delta_{Y_\tau}+B_r B_r^T \Delta_{Q_{r,\tau}}  \right)
\end{aligned}
\end{equation}
$\Delta_{Y_\tau}$, $\Delta_{Q_{r,\tau}}$, $\Delta_{e^{A_r\tau}}$  are the perturbations in $Y_{\tau}$, $Q_{r,\tau}$ and $e^{A_r\tau}$ respectively due to the perturbation $\Delta_{A_r}$ in $A_r$. 
The relation between the perturbations $\Delta_{Y_\tau}$ and $\Delta_{A_r}$ is through equation (\ref{Observability Sylvester Equation}).
\begin{equation}\label{Perturbation in Ytau}
    A^T\Delta_{Y_\tau}+\Delta_{Y_\tau}A_r+Y_\tau \Delta_{A_r}+e^{A^T\tau}C^{T}C_r\Delta_{e^{A_r \tau}} = 0
\end{equation}
Similarly, the relation between the perturbation $\Delta_{Q_{r,\tau}}$ and $\Delta_{A_r}$ is through equation (\ref{Reduced Observability Lyapunov Equation}).
\begin{equation}\label{Perturbation in Prtau}
\begin{aligned}
    A_r^T\Delta_{Q_{r,\tau}}+\Delta_{Q_{r,\tau}}A_r+\Delta_{A_r}^T Q_{r,\tau}+Q_{r,\tau}\Delta_{A_r}- \\\Delta_{e^{A_r\tau}}^T C_r^T C_r e^{A_r \tau}-
    e^{A_r^T}C_r^T C_r \Delta_{e^{A_r \tau}} = 0
\end{aligned}
\end{equation}
Applying Lemma \ref{Lemma1} for equations (\ref{Extra Sylvester Equation}) and (\ref{Perturbation in Ytau}) we get,
\begin{equation}\label{Differential J Term 1}
    \text{Tr}\left( Y_\tau \Delta_{A_r}X^T + e^{A^T\tau}C^T C_r\Delta_{e^{A_r \tau}} X^T \right) = \text{Tr}\left( B_r B^T \Delta_{Y_\tau}\right)
\end{equation}
Similarly considering (\ref{Extra Lyapunov Equation}) and (\ref{Perturbation in Prtau}) and using Lemma \ref{Lemma1} we get,
\begin{equation}\label{Differential J Term 2}
\begin{aligned}
    \text{Tr}( \Delta_{A_r}^T Q_{r,\tau}P_r+Q_{r,\tau}\Delta_{A_r}P_r-\Delta_{e^{A_r \tau}}^TC_r^TC_re^{A_r \tau}P_r- \\
    e^{A_r^T \tau}C_r^T C_r \Delta_{e^{A_r \tau}}P_r ) 
    = \text{Tr}\left( B_r B_r^T \Delta_{Q_{r,\tau}} \right) 
    \end{aligned}
\end{equation}
Using (\ref{Differential J Term 1}) and (\ref{Differential J Term 2}) in the expression (\ref{Differential J}) we get,
\begin{equation}\label{DeltaJAr}
\begin{split}
 \Delta_J^{A_r}&= 
           2\text{Tr}\left[ \left(X^T Y_\tau +P_r Q_{r,\tau} \right)\Delta_{A_r}\right]\\
          &+ 2\text{Tr}\left[ \left(X^Te^{A^T\tau}C^T C_r-P_r e^{A_r^T \tau}C_r^TC_r \right)\Delta_{e^{A_r \tau}} \right]
\end{split}
\end{equation}
From (\ref{Stau}), we get $2\text{ Tr}\left[ \left(X^Te^{A^T\tau}C^T C_r-P_r e^{A_r^T \tau}C_r^TC_r \right)\Delta_{e^{A_r \tau}} \right]=2\text{ Tr}\left[S_{\tau}\Delta_{e^{A_r \tau}} \right]$.\\
The  Fr{\'e}chet derivative of matrix exponential $f(A)=e^A$ along a perturbation matrix $E$ is defined as,
\begin{equation}\label{Frechet Derivative}
    L(A_r,E) = \int_{0}^{1} e^{A_r(1-s)}Ee^{A_r s}ds
\end{equation}
 Using the Fr{\'e}chet derivative expression (\ref{Frechet Derivative}) we can express  $2\text{ Tr}\left[S_{\tau}\Delta_{e^{A_r \tau}} \right]$ as follows,
\begin{equation}\label{FrechetExp}
\begin{split}
    2\text{ Tr}\left[S_{\tau}\Delta_{e^{A_r \tau}} \right] &= 
      2\tau \text{ Tr}\left[ L(A_r \tau,S_{\tau})\Delta_{A_r} \right]
\end{split}
\end{equation}
Using (\ref{FrechetExp}), we can rewrite ($\ref{DeltaJAr}$) as follows
\begin{equation}\label{DifferentialJA}
\begin{aligned}
   \Delta_J^{A_r} &= 2 \text{ Tr}\left[ \left(X^T Y_\tau+P_r Q_{r,\tau}\right)\Delta_{A_r}\right]+2\tau \text{Tr } \left(L(A_r\tau,S_\tau) \Delta_{A_r} \right) \\
     &= \langle 2({Y_\tau}^T X+Q_{r,\tau}P_r+\tau (L(A_r \tau,S_{\tau})^T) ), \Delta_{A_r}   \rangle
\end{aligned}
\end{equation}
Using the relation $\Delta_J^{A_r} = \langle \Delta_{A_r}J,\Delta_{A_r} \rangle$ and (\ref{DifferentialJA}), we obtain (\ref{GradA}).

 To get $\nabla_{B_r} J$, we perturb $B_r$ in the cost expression (\ref{ErrorCostFunction2}). The resulting first order perturbation is given by,
\begin{equation}
\begin{aligned}
    \Delta_J^{B_r} &= \text{Tr}\left( 2B^T Y_\tau \Delta_{B_r} + \Delta_{B_r}^{T}Q_{r,\tau}B_r+B_r^T Q_{r,\tau}\Delta_{B_r} \right) \\
    &= \langle 2\left( Y_\tau ^T B + Q_{r,\tau}B_r \right),\Delta_{B_r}\rangle
\end{aligned}
\end{equation}
Utilizing the relation $\Delta_J^{B_r} = \langle \nabla_{B_r}J, \Delta_{B_r} \rangle$, we have
(\ref{GradB}).

Consider the error cost expression (\ref{ErrorCostFunction1}). The first order perturbation in $J$ due to perturbation $\Delta_{C_r}$ of $C_r$ is,
\begin{equation}
\begin{aligned}
    \Delta_J^{C_r} &= \text{Tr}\left( -2C X_\tau \Delta_{C_r}^T + \Delta_{C_r}P_{r,\tau}C_r^T + C_r P_{r,\tau}\Delta_{C_r}^T \right) \\
             &= \text{Tr}\left( 2\left(C_r P_{r,\tau}-C X_\tau \right)^T \Delta_{C_r}  \right)\\
             &= \langle 2\left(C_r P_{r,\tau}-C X_\tau \right), \Delta_{C_r} \rangle
\end{aligned}    
\end{equation}
From the above expression we get $\nabla_{C_r} J$ (\ref{GradC}). 

\end{proof}


\begin{rem}\label{Remark 1}
Reference \cite{goyal2019time} derives Lyapunov based and \cite{sinani2019h2} derives interpolation based first-order necessary conditions for time-limited H$_2$ optimality. However, both these derivations require that $A_r$ should be diagonalizable. 
 Setting the gradients that we derived in Theorem \ref{Theorem 1} to zero gives us another set of Lyapunov based or Wilson's first-order necessary condition for H$_{2,\tau}$ optimality. We do not assume diagonalizability of $A_r$ for deriving the optimality conditions. It can be easily verified that our optimality conditions are equivalent to the Lyapunov based optimality conditions of \cite{goyal2019time} if $A_r$ is diagonalizable. Further, diagonalizability of $A_r$ also ensures that our optimality conditions are equivalent to the interpolation based H$_{2,\tau}$ optimality conditions derived in \cite{sinani2019h2}. We prove this in the Appendix. 
\end{rem}
\subsection{A Numerical method for \text{H}$_{2,\tau}$ model reduction}
\label{Algorithm}
The optimization problem (\ref{H2tau optimal model reduction}) considers $J=\left\Vert G-G_r\right\Vert_{\text{H}_{2,\tau}}^2$ as the objective function and $\{A_r,B_r,C_r\}$ as the optimization variables.  The optimization problem is nonlinear and non-convex. Hence finding global minimizers is difficult. However, we can use standard nonlinear optimization techniques with good initial conditions to obtain local minimizers \cite{nocedal2006numerical}.  

In this work, we solve the above optimization problem using standard quasi-Newton solvers by employing the MATLAB function 'fminunc'. The Hessian matrix is updated by means of the  Broyden-Fletcher-Goldfarb-Shanno (BFGS) algorithm . The gradients required for the BFGS algorithm are calculated using the closed form expressions  derived in Section \ref{Gradients} (Equations \ref{GradA}), (\ref{GradB}) and (\ref{GradC}). Due to the non-convex nature of the optimization problem, a good starting point is very necessary to solve the optimization problem. We use TL-BT and TL-TSIA to reduce the original model and use the reduced model to initialize the optimization problem. When the convergence criteria becomes less than a preset error tolerance, the iterations are stopped. We name the proposed time-limited model reduction method as TL-H$_2$Opt.     

\begin{rem} \label{Remark 2}
From Remark \ref{Remark 1}, we note the equivalence of the Lyapunov based and interpolation based frameworks of optimality conditions. Thus, for a fixed finite-time interval $\begin{bmatrix}0 & \tau \end{bmatrix}$ we require a minimum of $r(m+p)$ parameters for representing the H$_{2,\tau}$ optimality conditions. This is similar to the infinite interval case \cite{van2008h2}. We have $r^2+r(m+p)$ parameters in our optimization problem, which leads to overparametrization. However, this does not impede obtaining better H$_{2,\tau}$ optimal models due to the nature of the quasi-Newton solvers used for solving the optimization problem as observed in \cite{petersson2014model, mckelvey1997system}. 
\end{rem}

\subsection{Computational Cost} 
\label{Computation Cost}
 We now discuss the computational cost of the proposed method, TL-H$_2$Opt.  
Computation of $P_\tau$ and $Q_\tau$ requires solving the Lyapunov equations (\ref{Full Controllability Lyapunov Equation}) and (\ref{Full Observability Lyapunov Equation}) respectively. The equations are solved by the MATLAB function 'lyap'. The underlying algorithm for 'lyap' has a computational complexity of $\mathcal{O}(n^3)$. The computation of these quantities is costly. However, they are independent of optimization variables ($\{A_r,B_r,C_r\}$) and hence need to be computed only once before the start of the optimization process. The terms $P_{r,\tau}$ and $Q_{r,\tau}$ are dependent on the optimization variables and need to be computed at every iteration of the optimization process. Both these terms have computing cost $\mathcal{O}(n_r^3)$. Since $n_r \ll n$, the reduced order gramians are not computationally heavy. The exponential term $e^{A\tau}$ is computed with the MATLAB function 'expm' and has a high computation complexity of $\mathcal{O}(n^3)$. However, this term needs to be computed only once since it doesn't involve the optimization variables. The terms $e^{A_r \tau}$ and $L(e^{A_r \tau}, S_{\tau})$ include optimization variables and need to be computed at every iteration of the optimization process. For computing these terms, we use Algorithm 3 of \cite{al2009computing} which has a computational cost of $\mathcal{O}(n_r^3)$. 

The terms $Y_{\tau}$, $X_{\tau}$ and $X$ are required for calculating the cost function as well as the gradients and have to be computed at every iteration of the optimization problem. They are solutions of the Sylvester equations (\ref{Observability Sylvester Equation}), (\ref{Controllability Sylvester Equation}) and (\ref{Extra Sylvester Equation}) respectively. Computing them with the 'lyap' function in MATLAB costs $\mathcal{O}(n^3)$. This method of computing $Y_{\tau}$ and $X$ works for medium scale systems (order $<$ 1000) but becomes computationally expensive for large-scale models (order $>$ 1000). We can speed up the computations by using Algorithm 3 of \cite{benner2011sparse} to compute the Sylvester matrices. In this case, the cost of solving the Sylvester equations is much less than $\mathcal{O}(n^3)$ if the matrix $A$ is diagonal or has some sparse structure.   

\section{NUMERICAL EXAMPLES}
\label{Example}
In this section, we investigate the performance of the proposed algorithm TL-H$_2$Opt using two numerical examples. The first example is a SISO model of a beam with order 348. The second example  is a MIMO model of the International Space Station (ISS) with three inputs and three outputs and order 270. The examples are taken from \cite{chahlaoui2005benchmark}. The simulations are done in MATLAB version 8.3.0.532(R2014a) on a Intel(R) Core(TM) i5-6500 CPU @ 3.20GHz 3.19 GHz system with 16 GB RAM. We reduce the models over fixed finite time intervals. Using TL-TSIA and TL-BT, reduced models are obtained. These reduced models are further used to initialize TL-H$_2$Opt. The improvement in performance is assessed using the quantity $\Delta \text{Err}(\%)$ defined as
\begin{equation}\label{Improv}
\Delta \text{Err}(\%)=\frac{\text{Err\textsubscript{Alg}-Err\textsubscript{Opt}}}{\text{Err\textsubscript{Alg}}} \times 100 \%
\end{equation}
where Err\textsubscript{Alg} is the H$_{2,\tau}$ approximation error obtained by the algorithm Alg and  Err\textsubscript{Opt} is the error obtained by TL-H$_2$Opt with Alg initialization. The algorithm Alg may refer to either TL-TSIA or TL-BT in our case.
\subsection{Beam Example}
The first example is a clamped beam model of order 348 with one input and one output. We fix a time interval of $\begin{bmatrix} 0 & 1 \end{bmatrix}$. For this time interval, we use algorithms TL-TSIA and TL-BT to obtain reduced order approximations with $r$ varying from $r=2$ to $r=21$ increasing the value of $r$ one at a time. We use these low order models to initialize the TL-H$_2$Opt algorithm. The first subplot of Figure \ref{Figure 1} displays the H$_2(\tau)$ errors for the reduced models obtained by TL-TSIA and TL-H$_2$Opt with TL-TSIA initialization while the second displays the improvement in performance of the optimization based algorithm over TL-TSIA given by (\ref{Improv}) with Alg = TL-TSIA. Similarly, the first subplot of Figure \ref{Figure 2} compares the approximation errors due to TL-BT and TL-H$_2$Opt with TL-BT initialization and second shows the improvement in performance of TL-BT due to the time limited H$_2$ optimization algorithm. 

The reduced models obtained by TL-TSIA and TL-BT do not satisfy the H$_2(\tau)$ optimality conditions exactly[ref]. The reduced models obtained using TL-H$_2$Opt with TL-TSIA and TL-BT initialization improve the H$_2(\tau)$ approximation errors as evident from Figure \ref{Figure 1} and Figure \ref{Figure 2}. For $r=4, 18, 20$, the $H_2(\tau)$ approximation error due to TL-TSIA is high or it doesn't converge and hence errors corresponding to those orders are not included in the first subplot of Figure \ref{Figure 1}. For reduced models of order ($r$) less than $16$ some of the reduced orders show good improvement in the H$_2(\tau)$ errors; for instance in case of TL-TSIA $r=5,9$ and $10$ show an improvement of $79.82 \%$, $76.26 \%$ and $91.92 \%$ respectively. For TL-BT  $r=5,8,9$ and $14$ have their H$_2(\tau)$ errors reduced by $70.90 \%$, $74.74 \%$, $80.71 \%$ and $70.94 \%$ percent. Beyond order $r=16$, the optimization algorithm doesn't lead to any significant improvement in the H$_2(\tau)$ approximation errors for both TL-TSIA and TL-BT initialization. 

\begin{figure}[thpb]
\centering
\includegraphics[width=\columnwidth]{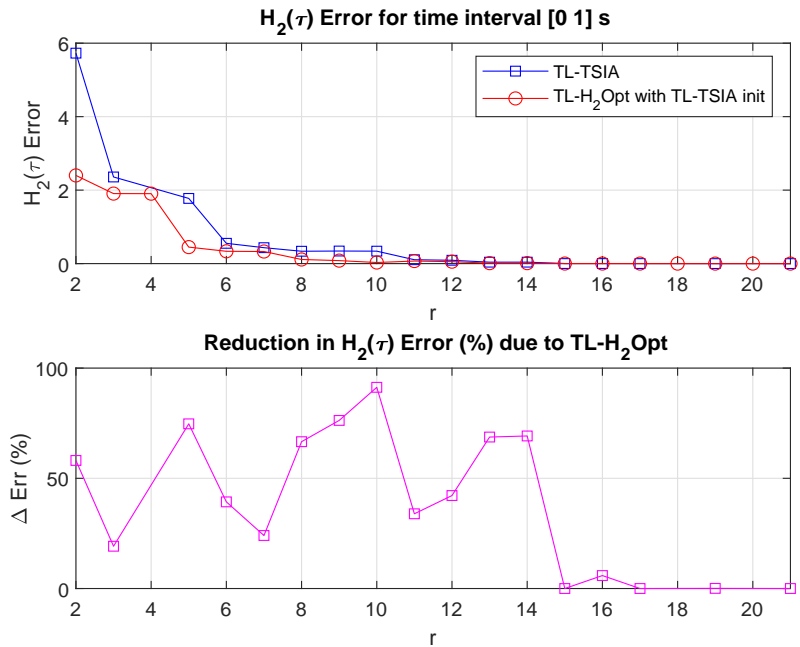}
\caption{Performance of TL-H$_2$Opt with TL-TSIA init for Beam Example}
\label{Figure 1}
\end{figure}

\begin{figure}[thpb]
\centering
\includegraphics[width=\columnwidth]{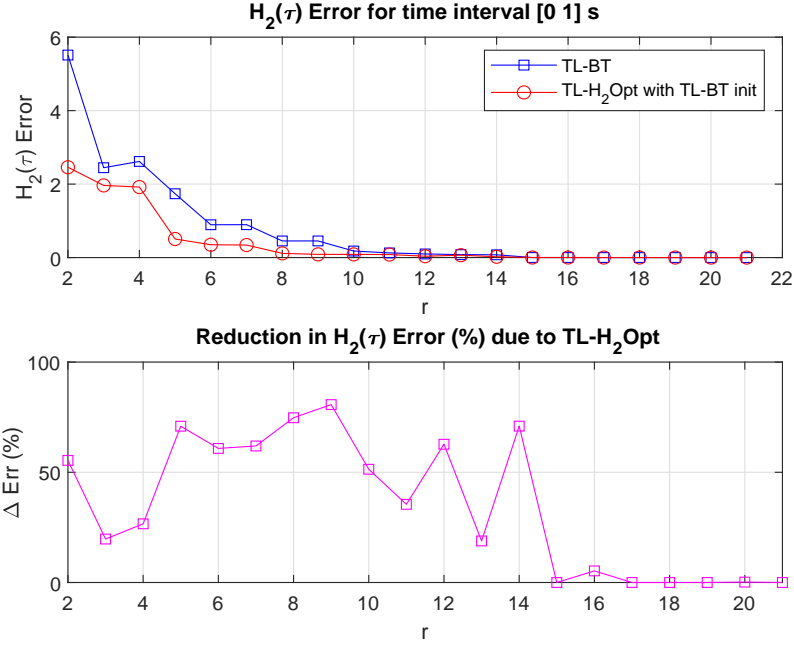}
\caption{Performance of TL-H$_2$Opt with TL-BT init for Beam Example}
\label{Figure 2}
\end{figure}

\subsection{ISS Example}
The second example is a MIMO model of the International Space Station (ISS) with three inputs and three outputs. For this example we consider a time interval of $\begin{bmatrix}0 & 0.5\end{bmatrix}$.

The H$_{2,\tau}$ approximation errors for the reduced models of various orders obtained with TL-TSIA and  TL-H$_2$Opt with TL-TSIA initialization along with performance improvement $\Delta \text{Err}(\%)$(\ref{Improv}) due to the optimization procedure is shown in Figure \ref{Figure 3}. Similar comparisons for TL-BT and TL-H$_2$Opt with TL-BT initialization are shown in Figure \ref{Figure 4}. We observe that the H$_{2,\tau}$ errors are improved considerably due to the application of TL-H$_2$Opt for reduced orders less than $r=38$ for both TL-TSIA and TL-BT initialization. There is no substantial improvement in the performance of TL-TSIA and TL-BT due to the application of TL-H$_2$Opt for reduced models of order greater than $r=38$.

The reduced models obtained by applying TL-TSIA for orders $r=4$ and $r=12$ have high H$_{2,\tau}$ approximation errors and are not shown in Figure \ref{Figure 3}. Initializing the optimization algorithm TL-H$_2$Opt by the reduced model obtained with TL-TSIA and solving the optimization problem improves the H$_{2,\tau}$ errors for all concerned reduced model orders ($r<38$) including $r=4$ and $r=12$ as evident from the second subplot of Figure \ref{Figure 3}. Apart from the case of $r=4$ and $r=12$ where the performance improvement is nearly $100\%$, the reduced order models with $r=22,26,28,30,32$ have performance improvement of over $70 \%$. Unlike the TL-TSIA initialization, there is no considerably high H$_{2,\tau}$ errors for any reduced order $r$ due to TL-BT initialization. Due to the application of the time-limited optimization algorithm with TL-BT reduced models as initial points, the H$_{2,\tau}$ error reduces as evident from Figure \ref{Figure 4}. The decrease in H$_{2,\tau}$ error for $r=18,20,22,24,26,28,32$ is more than $70 \%$.  

\begin{figure}[H]
\centering
\includegraphics[width=\columnwidth]{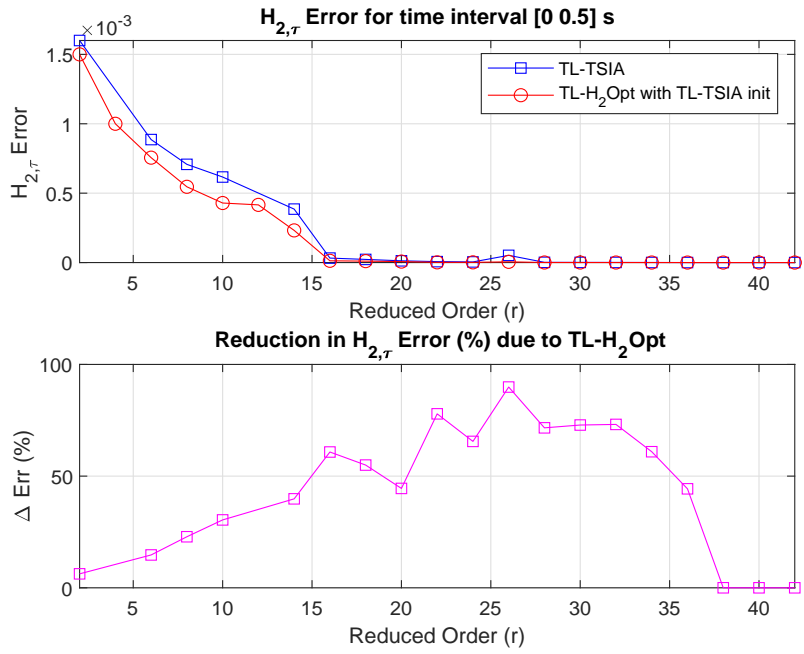}
\caption{Performance of TL-H$_2$Opt with TL-TSIA init for ISS Example}
\label{Figure 3}
\end{figure}

\begin{figure}[H]
\centering
\includegraphics[width=\columnwidth]{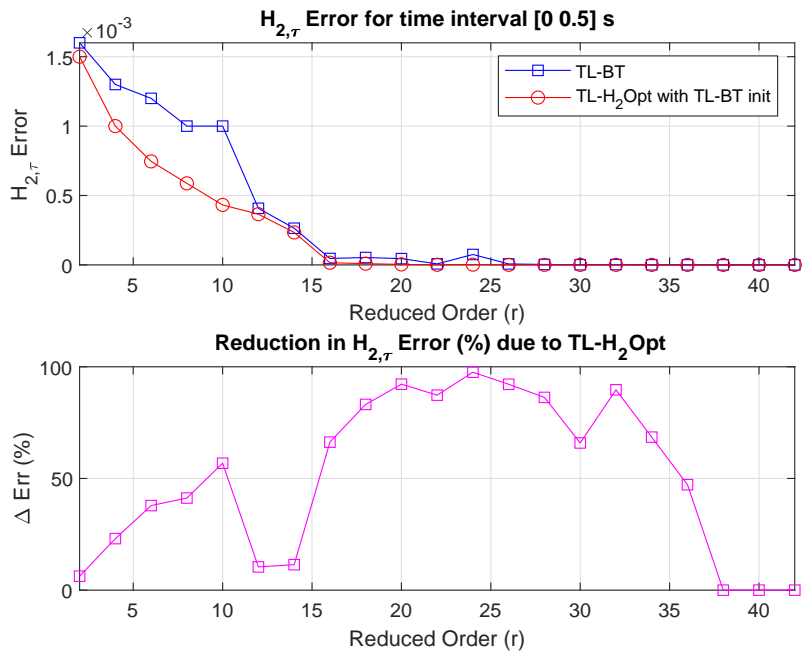}
\caption{Performance of TL-H$_2$Opt with TL-BT init for ISS Example}
\label{Figure 4}
\end{figure}

\section{Conclusion} \label{Conclusion}
In this work, we propose an optimization based method to obtain H$_{2,\tau}$ optimal reduced models. We derive closed form expressions of the gradients of an objective function defined over a limited time interval. The gradients are used with standard optimization algorithms for minimizing the objective function. The model reduction method proposed involves two steps. We first obtain a reduced model via TL-TSIA or TL-BT and then use the reduced model to initialize the optimization algorithm. Through a numerical example, we demonstrate the superiority of the TL-H$_2$Opt algorithm over TL-TSIA and TL-BT in obtaining better H$_{2,\tau}$ optimal reduced models.


\bibliographystyle{plain}
\bibliography{references}

\section*{Appendix}


\textbf{Equivalence of Lyapunov and Tangential Interpolation based H$_2(\tau)$ Optimality conditions}
    
The partial fraction expansion of the transfer function matrix $H_r(s)$ of system (\ref{eq2}) is given by
\begin{equation}\label{PartialFrac}
H_r(s) = \sum_{i=1}^{r} \frac{c_i b_i}{s-\lambda_i}
\end{equation}
where $\lambda_i \in \mathbb{C}$, $c_i \in \mathbb{C}^{p \times 1}$, $b_i \in \mathbb{C}^{1\times m}$ and $(\lambda_i,b_i,c_i)$ for $i=1,2,\hdots,r$ is self conjugate. Let, $\lambda_i^*$, $c_i^T$ and $b_i^T$ be the conjugate transpose of $\lambda_i$, $c_i$ and $b_i$ respectively. Let $v_i$ and $w_i$ be the right and left eigenvector respectively of the matrix $A_r$ corresponding to the eigenvalue $\lambda_i$. The following relations hold. 
\begin{equation}
    A_r v_i = \lambda_i v_i, C_r v_i = c_i, w_i^{T}A_r = \lambda_i w_i^{T}, w_i^{T}B_r = b_i
\end{equation}
For the transfer function $H(s)$ of system (\ref{eq1}), the time-limited counterpart over the time-interval $\begin{bmatrix}0 & \tau \end{bmatrix}$ is given by \cite{sinani2019h2} as,
\begin{equation}\label{time-limited transfer function}
    H_{\tau}(s) = C (sI_n-A)^{-1}(I_n-e^{-s\tau}e^{A\tau})B
\end{equation}
$H_{r,\tau}(s)$ is defined similarly.
We denote the identity matrix of size $n$ and $r$ as $I_n$ and $I_r$ respectively and
define the matrix $V = \begin{bmatrix}v_1 & v_2 & \hdots & v_r.  \end{bmatrix}$. 
\begin{theorem} \label{Theorem 2}
Let $H_r(s)$ given by (\ref{PartialFrac}) have $r$ distinct first order poles. $H_{\tau}(s)$ and $H_{r,\tau}(s)$ are the time limited transfer functions of $H(s)$ and $H_r(s)$ respectively over the time interval [0 $\tau$]. Then for $i = 1,2,\hdots,r$, $j=1,2,\hdots,r$ and $i \neq j$
\begin{subequations}
\begin{equation}\label{RightTangnetialError}
    \frac{1}{2}\left( \nabla_{B_r}J  \right)^T v_i = \left[ H_{r,\tau}^T(-\lambda_i^*)-H_{\tau}^T(-\lambda_i^*) \right]c_i
\end{equation}
\begin{equation}\label{LeftTangentialError}
     \frac{1}{2}w_i^T \left( \nabla_{C_r} J \right)^T = b_i \left[ H_{r,\tau}^T(-\lambda_i^*)-H_{\tau}^T(-\lambda_i^*) \right]
\end{equation}
\begin{equation}\label{BiTangentialError}
     \frac{1}{2}w_i^T\left( \nabla_{A_r}J \right)^T v_i = b_i \left. \frac{d}{ds}\left[ H_{\tau}^T(s)- H_{r,\tau}^T(s)\right] \right\vert_{s=-\lambda_i^*}c_i
\end{equation}
\begin{equation}\label{OverParametrizationError}
    \frac{1}{2}w_i^T \left( \nabla_{A_r}J \right)^T v_j = \frac{1}{2(\lambda_i-\lambda_j)}\left[ b_i\left(\nabla_{B_r} J \right)^T v_j -w_i^T \left( \nabla_{C_r} J \right)^T c_j \right]
\end{equation}
\end{subequations}
\end{theorem}

\begin{proof}
Let us define $x_{i,\tau}=X_\tau w_i, x_i=X w_i,  p_{i,\tau}=P_{r,\tau}w_i, p_i=P_r w_i, y_{i,\tau} = Y_\tau v_i, q_{i,\tau}= Q_{r,\tau}v_i$. Using (\ref{Full Controllability Lyapunov Equation}), (\ref{Controllability Sylvester Equation}), (\ref{Reduced Controllability Lyapunov Equation}), (\ref{Full Observability Lyapunov Equation}), (\ref{Observability Sylvester Equation}) and (\ref{Reduced Observability Lyapunov Equation}) we get
\begin{subequations}
\begin{equation}\label{xitau}
    x_{i,\tau} = -(A+\lambda_i^* I_n)^{-1}(I_n-e^{\lambda_i^* \tau}e^{A\tau})B b_i^T
\end{equation}
\begin{equation}\label{xi}
    x_i = -(A+\lambda_i^* I_n)^{-1}B b_i^T
\end{equation}
\begin{equation}\label{pitau}
    p_{i,\tau} = -(A_r+\lambda_i^* I_r )^{-1}(I_r-e^{\lambda_i^* \tau}e^{A_r \tau})B b_i^T
\end{equation}
\begin{equation}\label{pi}
    p_i = -(A_r+\lambda_i^* I_r )^{-1}B_r b_i^T
\end{equation}
\begin{equation}\label{yitau}
    y_{i,\tau} = (A^T+\lambda_i I_n )^{-1}(I_n-e^{\lambda_i \tau}e^{A^T \tau})C^T c_i
\end{equation}
\begin{equation}\label{qitau}
    q_{i,\tau} = -(A_r^T+\lambda_i I_r)^{-1}(I_r-e^{\lambda_i \tau}e^{A_r^T \tau})C_r^T c_i
\end{equation}
\end{subequations}
We use (\ref{qitau}) and (\ref{yitau}) and obtain (\ref{RightTangnetialError}) as follows,

\begin{equation}
\begin{aligned}
    \frac{1}{2}\left(\nabla_{B_r}J \right)^T v_i &=  B_r^T q_{i,\tau}+B^T y_{i,\tau}\\
    &= -B_r^T (A_r^T+\lambda_i I_r)^{-1}(I_r-e^{\lambda_i \tau}e^{A_r^T \tau})C_r^T c_i+ \\
    &B^T(A^T+\lambda_i I_n)^{-1}(I_n-e^{\lambda_i \tau}e^{A^T \tau})C^T c_i \\
    &= [H_{r,\tau}^T(-\lambda_i^*)-H_{\tau}^T(-\lambda_i^*)]c_i
\end{aligned}    
\end{equation}
Similarly, using (\ref{pitau}) and (\ref{xitau}) the second expression becomes
\begin{equation}
\begin{aligned}
    \frac{1}{2}(\nabla_{C_r}J)w_i &= 
     [H_{r,\tau}(-\lambda_i^*)-H_{\tau}(-\lambda_i^*)]b_i^T
\end{aligned}
\end{equation}
The third expression is derived as follows.
\begin{equation}\label{thirdexp}
\begin{aligned}
&\frac{1}{2} w_i^T \left( \nabla_{A_r} J \right)^T v_j  \\
&= x_i^T y_{j,\tau}+p_i^T q_{j,\tau}+ w_i^T S_{\tau}v_j \int_{0}^{1} e^{\lambda_i(\tau-\tau s)}e^{\lambda_j \tau s}(\tau ds)
\end{aligned}
\end{equation}
Using (\ref{time-limited transfer function}), we have $\frac{d}{ds} H_{\tau}(s) = -C(s I_n-A)^{-2}(I_n-e^{-s\tau}e^{A\tau})B+ \tau e^{-s\tau}C(s I_n-A)^{-1}e^{A\tau}B$ and
$\frac{d}{ds} H_{r,\tau}(s) = -C_r(s I_r-A_r)^{-2}(I_r-e^{-s\tau}e^{A_r\tau})B_r+ \tau e^{-s\tau}C_r(s I_r-A_r)^{-1}e^{A_r\tau}B_r$. Substituting $S_{\tau}$ for $i=j$, the expression (\ref{thirdexp}) reduces to 
\begin{equation}\label{thirdexppart1}
\begin{aligned}
&x_i^T y_{i,\tau} + p_i^T q_{i,\tau} + \\
&(x_i^T e^{A^T \tau}C^Tc_i-p_i^Te^{A_r^T\tau}C_r^T c_i)\int_{0}^{1} e^{\lambda_i(\tau-\tau s)}e^{\lambda_i \tau s} \tau ds 
\end{aligned}
\end{equation}
Using (\ref{xi}) and (\ref{yitau}), the term $x_i^T y_{i,\tau}$ becomes $-b_i B^T(A^T+\lambda_i I_n)^{-2}(I_n-e^{\lambda_i \tau}e^{A^T \tau})C^T c_i$, using (\ref{pitau}) and (\ref{qitau}) the term $p_i^T q_{i,\tau}^T$ becomes $b_i B_r^T (A_r^T+ \lambda_i I_r)^{-2}(I_r-e^{\lambda_i \tau}e^{A_r^T \tau})C_r^T c_i $ and after substituting $x_i$ and $p_i$, the third term in the expression (\ref{thirdexppart1}) becomes $-\tau e^{\lambda_i \tau} b_i B^T(A^T+\lambda_i I_n)^{-1}e^{A^T \tau}C^T c_i+ \tau e^{\lambda_i \tau} b_i B_r^T(A_r^T + \lambda_i I_r)^{-1}e^{A_r^T \tau}C_r^T c_i$. Combining the three terms we get 
(\ref{BiTangentialError}).\\
For the case $i \neq j$, the expression (\ref{thirdexp}) becomes
\begin{equation}\label{thirdexppart2}
\begin{aligned}
   & \frac{1}{2}w_i^T \left( \nabla_{A_r} J \right)^T v_j \\
   &= x_i^T y_{j,\tau} + p_i^T q_{j,\tau}+w_i^T S_{\tau}v_j \left( \frac{e^{\lambda_i \tau}-e^{\lambda_j \tau}}{\lambda_i-\lambda_j} \right)
\end{aligned} 
\end{equation}
Using the identity $\left(A^T+\lambda_i I_n\right)^{-1}\left(A^T+\lambda_j I_n \right)^{-1} = \frac{1}{\lambda_i-\lambda_j}\left[ \left(A^T+\lambda_j I_n \right)^{-1} - \left(A^T+\lambda_i I_n \right)^{-1} \right]$ and (\ref{xi}), (\ref{yitau}) , the term $x_i^T y_{j,\tau}$ becomes
\begin{equation}\label{Expression1}
\begin{aligned}
-\frac{-1}{\lambda_i-\lambda_j}b_i B^T\left[ \left(A^T+\lambda_j I_n \right)^{-1} - \left(A^T+\lambda_i I_n \right)^{-1} \right]\\
\left(I_n-e^{\lambda_j \tau}e^{A^T \tau}\right)C^T c_j
\end{aligned}
\end{equation}
Similarly, the term $p_i^T q_{j,\tau}$ becomes 
\begin{equation}\label{Expression2}
\begin{aligned}
  \frac{-1}{\lambda_i-\lambda_j}b_i B_r^T\left[ \left(A_r^T+\lambda_j I_r \right)^{-1} - \left(A_r^T+\lambda_i I_r \right)^{-1} \right]\\
  \left(I_r-e^{\lambda_j \tau}e^{A_r^T \tau}\right)C_r^T c_j
\end{aligned}  
\end{equation}
The third right hand side term of (\ref{thirdexppart2}) becomes 
\begin{equation}\label{Expression3}
\begin{aligned}
    \frac{e^{\lambda_i \tau}-e^{\lambda_j \tau}}{\lambda_i-\lambda_j}\left[\left( -b_i B^T(A^T+\lambda_i I_n)^{-1}e^{A^T \tau} C^T c_j \right)\right]+ \\
  \frac{e^{\lambda_i \tau}-e^{\lambda_j \tau}}{\lambda_i-\lambda_j}  \left[ \left( b_i B_r^T(A_r^T+\lambda_i I_r)^{-1}e^{A_r^T \tau} C_r^T c_j \right)\right]
\end{aligned}
\end{equation}
Adding the expressions (\ref{Expression1}), (\ref{Expression2}) and (\ref{Expression3}), we obtain (\ref{OverParametrizationError}).
\end{proof}

The above theorem shows that setting  diag $V^{-1}\left(\nabla_{A_r} J \right)^T V$, $\left( \nabla_{B_r} J \right)^T$ and $\left( \nabla_{C_r} J \right)^T$ to 0 gives the interpolation based H$_{2,\tau}$ optimality conditions.This proves that Lyapunov based and interpolation based optimality conditions are equivalent.

\end{document}